\documentclass{llncs}

\usepackage{amssymb}
\usepackage{amsmath}
\usepackage{hyperref}
\usepackage{graphicx}
\usepackage{lineno}
\usepackage{subfig}

\providecommand{\pb}[1]{{\sc #1} problem}
\providecommand{\pbs}[1]{$\mathsf{#1}$ problem}
\providecommand{\seacabo}{\hfill $\Box$}

\title{Matching colored points with rectangles\thanks{This research has been partially supported by grant CONICYT, FONDECYT/Iniciaci\'on 11110069 (Chile).}}

\author{
L. E. Caraballo\inst{1}
\and 
C. Ochoa\inst{2}
\and
P. P\'{e}rez-Lantero\inst{3}
\and
J. Rojas-Ledesma\inst{2}
}

\institute{
	Facultad de Matem\'atica y Computaci\'on, Universidad de La Habana, Cuba. \email{luis.caraballo@iris.uh.cu} \and
	Departamento de Ciencias de la Computaci\'on, Universidad de Chile, Chile. \email{\{cochoa,jrojas\}@dcc.uchile.cl} \and
	Escuela de Ingenier\'ia Civil en Inform\'atica, Universidad de Valpara\'{i}so, Chile. \email{pablo.perez@uv.cl}
}

\begin{document}
\maketitle
\linenumbers

\begin{abstract}
Let $S$ be a point set in the plane
such that each of its elements is colored either red or blue. 
A matching of $S$ with rectangles is any set of pairwise-disjoint 
axis-aligned rectangles such that each rectangle contains 
exactly two points of $S$.
Such a matching is monochromatic if every rectangle contains
points of the same color, and is bichromatic if every rectangle
contains points of different colors.
In this paper we study the following two problems:
\begin{enumerate}
\item Find a maximum monochromatic matching of $S$ with rectangles.
\item Find a maximum bichromatic matching of $S$ with rectangles.
\end{enumerate}
For each problem we provide a polynomial-time approximation
algorithm that constructs a matching with at least $1/4$ of the
number of rectangles of an optimal matching.
We show that the first problem is $\mathsf{NP}$-hard even 
if either the matching rectangles
are restricted to axis-aligned segments or $S$ is in general position,
that is, no two points of $S$ share
the same $x$ or $y$ coordinate.
We further show that the second problem is also $\mathsf{NP}$-hard,
even if $S$ is in general position.
These $\mathsf{NP}$-hardness results follow by showing that deciding
the existence of a perfect matching is $\mathsf{NP}$-complete in each case.
The approximation results are based on a relation of our problem
with the problem of finding a maximum independent set in a family
of axis-aligned rectangles. 
With this paper we extend previous ones on matching one-colored points with
rectangles and squares, and matching two-colored points with segments.
Furthermore, using our techniques, we prove that it is $\mathsf{NP}$-complete
to decide a perfect matching with rectangles in the case where all points have
the same color, solving an open problem of Bereg, Mutsanas, and Wolff [CGTA (2009)].
\end{abstract}

\section{Introduction}\label{sec:intro}

Matching points in the plane with geometric objects consists in, given an input point set 
$S$ and a class $\mathcal{C}$ of geometric objects, finding a collection $M\subseteq \mathcal{C}$
such that each element of $M$ contains exactly two points of $S$ and every point of $S$
lies in at most one element of $M$. This kind of geometric matching problem was introduced
by \'Abrego et al.~\cite{abrego2009}, calling a geometric matching {\em strong}
if the geometric objects are disjoint, and {\em perfect} if every point of
$S$ belongs to some element of $M$. They studied the existence and properties of 
matchings for point sets in the plane when $\mathcal{C}$ is the set of axis-aligned squares,
or the family of disks.

Bereg et al.~\cite{bereg2009} continued the study of this class of problems. They 
proved by a constructive proof that if $\mathcal{C}$ is the class of axis-aligned rectangles, 
then every point set of $n$ points in the plane admits a strong matching that matches at least
$2\lfloor n/3\rfloor$ of the points; and leaved open the computational complexity of
finding such a maximum strong matching. They assume that there can be points with
the same $x$ or $y$ coordinate, condition that makes the optimization problem hard.  
In the case in which $\mathcal{C}$ is the class of axis-aligned 
squares, they proved that it is $\mathsf{NP}$-hard to decide whether a given point set admits a 
perfect strong matching.

In the setting of colored points, it is well known that every two-colored point set in the
plane such that no three points are collinear, consisting of $n$ red points and $n$ blue points, admits a perfect 
strong matching with straight segments, where
each segment connects points of different colors~\cite{larson1990}. 
Dumitrescu and Steiger~\cite{Dumitrescu2000} introduced
the study of strong straight segment matchings of two-colored point sets in the case where each segment must
match points of the same color. The current results are due to
Dumitrescu and Kaye~\cite{Dumitrescu2001}: Every two-colored point set $S$ of $n$ points
admits a strong straight segment matching that matches at least $\frac{6}{7}n-O(1)$ of the points,
which can be found in $O(n^2)$ time; and there exist $n$-point 
sets such that every strong matching with straight segments matches
at most $\frac{94}{95}n+O(1)$ points. The computational complexity of deciding if a 
given two-colored point set admits a perfect strong matching with straight segments 
connecting points of the same color, is
still an open problem~\cite{Dumitrescu2000}. 

Let $S=R\cup B$ be a set of $n$ points in the plane such that each element of
$S$ is colored either red or blue, where $R$ denotes the set of the points
colored red and $B$ the set of the points colored blue. A strong matching of $S$ is
{\em monochromatic} if all matching objects
cover points of the same color. Likewise, a strong matching of $S$ is
{\em bichromatic} if all matching objects
cover points of different colors.

As an extension of the above problems, we study both monochromatic and bichromatic 
strong matchings of $S$ with axis-aligned rectangles. For the
monochromatic case it is trivial to build examples in which no matching rectangle exists 
and examples in which a perfect strong matching exists.
%
For the bichromatic case, there always exists at least one matching rectangle 
(i.e.\ match the red point and the blue point such that their minimum enclosing rectangle
has minimum area among all combinations of a red point and a blue point)
and similar as the the monochromatic case, one can build examples in which exactly one
matching rectangle exists and examples in which a perfect strong matching exists.
Then, we focus our attention in the following optimization problems: 

\medskip

\noindent\pb{Maximum Monochromatic Rectangle Matching ($\mathsf{MMRM}$)}: {\em Find a monochromatic 
strong matching of $S$ with the maximum number of rectangles.}

\medskip

\noindent\pb{Maximum Bichromatic Rectangle Matching ($\mathsf{MBRM}$)}: {\em Find a bichromatic 
strong matching of $S$ with the maximum number of rectangles.}

\paragraph{Results.} 
For each problem we provide a polynomial-time approximation
algorithm that constructs a matching with at least $1/4$ of the
number of rectangles of an optimal matching.
In the approximation algorithms we consider that the elements of $S$
are not necessarily in general position. We say that $S$ is in 
{\em general position} if no two elements of $S$ share the same $x$ or
$y$ coordinate.
We further use the direct relation of the 
problems with the \pb{Maximum Independent Set of Rectangles},
which is to find a maximum subset of pairwise disjoint rectangles in a given family
of rectangles.
We complement the approximation results by
showing that the \pbs{MMRM} is $\mathsf{NP}$-hard, even 
if either the matching rectangles
are restricted to axis-aligned segments or the points
are in general position.
We further show that the \pbs{MBRM} is also $\mathsf{NP}$-hard,
even is the points are in general position.
Furthermore, we are able to prove that if all elements of $S$ have
the same color, then the \pbs{MMRM} keeps $\mathsf{NP}$-hard, solving
an open question of Bereg et al.~\cite{bereg2009}.
These $\mathsf{NP}$-hardness results follow by showing that deciding
the existence of a perfect matching is $\mathsf{NP}$-complete in each case.

\section{Preliminaries} 

For every point $p$ of $S$, let $x(p)$, $y(p)$, and $c(p)$ denote
the $x$-coordinate, the $y$-coordinate, and the color of $p$, respectively.
Given two points $a$ and $b$ of the plane with $x(a)\leq x(b)$, let $D(a,b)$ denote 
the rectangle which has the segment 
connecting $a$ and $b$ as diagonal, which is in fact the minimum enclosing axis-aligned
rectangle of $a$ and $b$. If $a$ and $b$ are horizontally or vertically aligned,
we say that $D(a,b)$ is a {\em segment}, otherwise we say that $D(a,b)$ is a {\em box}.
We say that $D(a,b)$ is {\em red} if both $a$ and $b$ are colored red. Otherwise,
if both $a$ and $b$ are colored blue, we say that $D(a,b)$ is {\em blue}.
Given $S$, consider the following two families of axis-aligned rectangles:
\begin{eqnarray*}
\mathcal{R}(S) & := & \bigl\{D(p,q)~|~p,q\in S; c(p)=c(q); \text{ and } D(p,q)\cap S=\{p,q\}\bigr\}\\
\overline{\mathcal{R}}(S) & := & \bigl\{D(p,q)~|~p,q\in S; c(p)\neq c(q); \text{ and } D(p,q)\cap S=\{p,q\}\bigr\}
\end{eqnarray*}
Observe that the \pbs{MMRM} is equivalent to finding a maximum subset of $\mathcal{R}(S)$
of independent rectangles. Two rectangles are {\em independent} if they are disjoint.
Similarly, the \pbs{MBRM} is equivalent to finding a maximum subset 
of $\overline{\mathcal{R}}(S)$ of independent rectangles.
The \pb{Maximum Independent Set of Rectangles ($\mathsf{MISR}$)} is a classical $\mathsf{NP}$-hard problem
in computational geometry and combinatorics, and is to
find a maximum subset of independent rectangles in a given 
set of axis-aligned rectangles~\cite{adamaszekW13,Agarwal200683,Chalermsook2011,Chalermsook2009,Fowler1981,Imai1983,rim1995}.
The general \pbs{MISR} admits a polynomial-time approximation algorithm,
which with high probability produces and independent set of rectangles with at least
$\Omega(\frac{1}{\log\log m})$ times the number of rectangles in an optimal solution, being
$m$ the number of rectangles in the input~\cite{Chalermsook2011,Chalermsook2009}.
There also exist deterministic polynomial-time
$\Omega(\frac{1}{\log m})$-approximation algorithms for the \pbs{MISR}~\cite{AgarwalKS98,KhannaMP98}.
Finding a constant-approximation algorithm, or a $\mathsf{PTAS}$, is still
an intriguing open question.
As we will show later, our two matching problems, being special cases of the \pbs{MISR},
are also $\mathsf{NP}$-hard and we give a polynomial-time $1/4$-approximation algorithm for each of them.

There exists polynomial-time exact algorithms, constant-approximation algorithms, and $\mathsf{PTAS}$'s for special 
cases of the \pbs{MISR}, according to the intersection graph of the rectangles. The {\em intersection graph}
is the undirected graph with the rectangles of the input as vertices, and two
rectangles are adjacent if they are not independent. For any set $\mathcal{H}$ of rectangles,
let $G(\mathcal{H})$ denote the intersection graph of $\mathcal{H}$.
Given two rectangles $R_1$ and $R_2$,
we say that $R_2$ {\em pierces} $R_1$ if into the $x$-axis the orthogonal projection of $R_1$ 
contains the orthogonal projection of $R_2$, and into the $y$-axis the orthogonal projection of $R_2$ 
contains the orthogonal projection of $R_1$. 
We say that two intersecting rectangles {\em pierce} if 
one of them pierces the other one (see Figure~\ref{fig:pierce})~\cite{Agarwal200683,LewinEytan2004,soto2011}. 
Independently, Agarwal and Mustafa~\cite{Agarwal200683} and Lewin-Eytan et al.~\cite{LewinEytan2004}
showed that if all rectangles are pairwise-piercing then the 
\pbs{MISR} can be solved in polynomial time since in this case the intersection graph $G$
of the rectangles is perfect.
Using a classical result of Gr\"{o}tschel et al.~\cite{Grotschel1984325},
a maximum independent set of a perfect graph can be computed in polynomial time. 
Agarwal and Mustafa~\cite{Agarwal200683} generalized this fact, claiming that
the spanning subgraph of the intersection graph, with only the edges corresponding to the piercing intersections, 
is also perfect. We will use these results on pairwise-piercing rectangles
in our approximation algorithms.
If $q$ is the clique number of the intersection graph, there
exists a $(1/4q)$-approximation~\cite{Agarwal200683,LewinEytan2004}. In our two 
problems, we can build examples in which the size of the optimal solution is either
big or small, and independently of that, the clique number $q$ is either big or small. 
Then, applying this result does not always give a good approximation.

According to the nature of the rectangles in our families $\mathcal{R}(S)$ and $\overline{\mathcal{R}}(S)$,
two rectangles can have one of four types of intersection:
(1) a {\em piercing} intersection in which the two rectangles pierce (see Figure~\ref{fig:pierce});
(2) a {\em corner} intersection in which each rectangle contains exactly one
of the corners of the other one and these corners are not elements of $S$ (see Figure~\ref{fig:corner});
(3) a {\em point} intersection where the intersection of the rectangles is precisely
an element of $S$ (see Figure~\ref{fig:point}); and
(4) a {\em side} intersection which is the complement of the above three intersection
types (see Figure~\ref{fig:side}). 

\begin{figure}[h]
	\centering
	\subfloat[]{
		\includegraphics[scale=0.6,page=1]{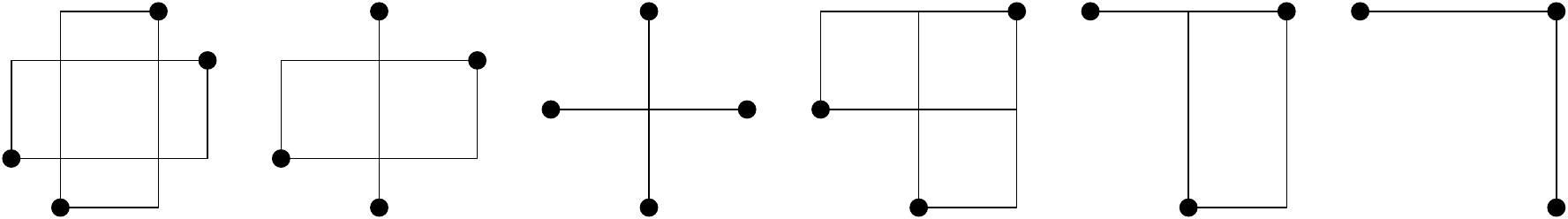}
		\label{fig:pierce}
	}\\
	\subfloat[]{
		\includegraphics[scale=0.6,page=2]{intersections.pdf}
		\label{fig:corner}
	}\hspace{0.3cm}
	\subfloat[]{
		\includegraphics[scale=0.6,page=3]{intersections.pdf}
		\label{fig:point}
	}
	\hspace{0.3cm}
	\subfloat[]{
		\includegraphics[scale=0.6,page=4]{intersections.pdf}
		\label{fig:side}
	}
	\caption{\small{The four types of intersection: (a) piercing. (b) corner.
	(c) point. (d) side.}}
	\label{fig:intersections}
\end{figure}

Let $G:=G(\mathcal{R}(S))$.
Observe that if we consider the spanning subgraph $G'$ of $G$ with edge set the edges corresponding to 
the piercing intersections, and compute in polynomial time the maximum independent set 
for $G'$~\cite{Agarwal200683,LewinEytan2004},
then we will obtain a set $H\subseteq\mathcal{R}(S)$ of pairwise non-piercing rectangles. In that
case the set $H$ (after a slight perturbation that maintains the same 
intersection graph) is a set of pseudo-disks and the $\mathsf{PTAS}$ of Chan and Har-Peled~\cite{chan2009},
for approximating the maximum independent set in a family of pseudo-disks,
can be applied in $H$ to obtain an independent set $H'\subseteq H \subseteq \mathcal{R}(S)$. Unfortunately, 
we are unable 
to compare $|H'|$ with the optimal value of the \pbs{MISR} for $\mathcal{R}(S)$. The same
arguments apply for $\overline{\mathcal{R}}(S)$. On the other hand, there exits $\mathsf{PTAS}$'s
for the \pbs{MISR} when the rectangles have unit height~\cite{Chan200419}, and
bounded aspect ratio~\cite{Chan03,ErlebachJS05}.

Soto and Telha~\cite{soto2011} studied the following problem to model cross-free matchings
in two-directional orthogonal ray graphs (2-dorgs): Given both a point set $A_1$ and a point set $A_2$,
find a maximum set of independent rectangles over all rectangles having 
an element of $A_1$ as bottom-left corner and an element of $A_2$ as top-right corner.
For $A_1:=R$ and $A_2:=B$, where $S=R\cup B$, this problem is equivalent to the \pbs{MISR} over the rectangles
$H\subseteq\overline{\mathcal{R}}(S)$ that have a red point as bottom-left corner and a blue point 
as top-right corner. 
The authors solved this problem in polynomial time with the next observations:
the rectangles of $H$ have only two types of intersections, piercing and corner, and $H$ can be reduced
to a small one $H'\subseteq H$ whose intersection graph is perfect since the elements of $H'$ 
are pairwise piercing, and a maximum
independent set in $H'$ is a maximum independent set in $H$.
They proved them by using an LP-relaxation approach. By using
simpler combinatorial arguments, we generalize and prove these observations 
to obtain our approximation algorithms.


\section{Approximation algorithms}\label{sec:approximation}

Given a point set $P$ in the plane, we say that $\mathcal{H}$ is a 
{\em set of rectangles on $P$} if every element of $\mathcal{H}$ is of the form
$D(a,b)$, where $a,b\in P$ and $D(a,b)$ contains exactly the points $a$ and $b$ of $P$.
We say that the set $\mathcal{H}$ is {\em complete} if for every pair of elements $D(a,b)$
and $D(a',b')$ of $\mathcal{H}$ that have a corner intersection, the other two rectangles
of the form $D(p,q)$ having a piercing intersection, 
where $p\in\{a,a'\}$ and $q\in\{b,b'\}$, also belong to $\mathcal{H}$
(see Figure~\ref{fig:rect-corner}).
Let $G_{p,c}(\mathcal{H})$ denote the spanning subgraph of $G(\mathcal{H})$ with edge set the edges
that correspond to the piercing and the corner intersections.

\begin{figure}[h]
	\centering
	\subfloat[]{
		\includegraphics[scale=0.5,page=1]{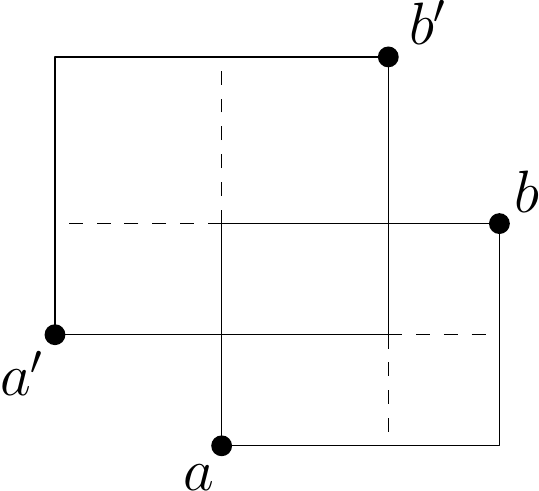}
		\label{fig:rect-corner}
	}\hspace{0.3cm}
	\subfloat[]{
		\includegraphics[scale=0.5,page=3]{discard_corners.pdf}
		\label{fig:case1}
	}\hspace{0.3cm}
	\subfloat[]{
		\includegraphics[scale=0.5,page=2]{discard_corners.pdf}
		\label{fig:case2}
	}	
	\caption{\small{(a) If for every $D(a,b),D(a',b')\in\mathcal{H}$ we have that
		$D(a,b'),D(a',b)\in\mathcal{H}$, then $\mathcal{H}$ is complete.
		(b,c) Cases in the proof of Lemma~\ref{lem:discard-corners}.}}
	\label{fig:remove-corners}
\end{figure}

\begin{lemma}\label{lem:discard-corners}
Let $P$ be a point set and $\mathcal{H}$ be any complete set
of rectangles on $P$. 
Let $D(a,b)$ and $D(a',b')$ be two elements of $\mathcal{H}$ such that
$D(a,b)$ and $D(a',b')$ have a corner intersection.
A maximum independent set in $G_{p,c}(\mathcal{H}\setminus\{D(a,b)\})$ is
a maximum independent set in $G_{p,c}(\mathcal{H})$. 
\end{lemma}

\begin{proof}
Let $I$ denote a maximum independent set of $G_{p,c}(\mathcal{H})$.
Assume w.l.o.g.\ that 
$x(a')<x(a)\leq x(b')<x(b)$ and $y(a)<y(a')\leq y(b)<y(b')$ (see Figure~\ref{fig:rect-corner}).
We claim that either $(I\setminus \{D(a,b)\})\cup\{D(a,b')\}$ or
$(I\setminus \{D(a,b)\})\cup\{D(a',b)\}$ is an independent set, which implies the result.
Indeed, if $(I\setminus \{D(a,b)\})\cup\{D(a,b')\}$ is an independent set, then we are done.
Otherwise, at least one of the next two cases is satisfied:
(1) there is a rectangle of $I\setminus\{D(a,b)\}$ that has a corner intersection with both 
$D(a',b')$ and $D(a,b')$ (see Figure~\ref{fig:case1}); and
(2)  there is a rectangle of $I\setminus\{D(a,b)\}$ that has a piercing intersection with both 
$D(a',b')$ and $D(a,b')$ (see Figure~\ref{fig:case2}).
In both cases $D(a',b)$ is independent from any rectangle
in $I\setminus\{D(a,b)\}$. Hence, $(I\setminus \{D(a,b)\})\cup\{D(a',b)\}$ is an independent set.\seacabo
\end{proof}

\begin{lemma}\label{lem:solve-pierce-misr}
Let $P$ be a point set and $\mathcal{H}$ be any complete set
of rectangles on $P$. A maximum independent set 
in $G_{p,c}(\mathcal{H})$ can be found in polynomial time.  
\end{lemma}

\begin{proof}
Using Lemma~\ref{lem:discard-corners}, the set $\mathcal{H}$ can be
reduced in polynomial time to the set $\mathcal{H}'\subseteq \mathcal{H}$ such 
that all edges of the graph $G_{p,c}(\mathcal{H}')$ correspond to
piercing intersections, and a maximum independent set in $G_{p,c}(\mathcal{H}')$
is a maximum independent set in $G_{p,c}(\mathcal{H})$.
The former one can be found in polynomial time since $G_{p,c}(\mathcal{H}')$ is a
perfect graph, precisely a comparability graph~\cite{Agarwal200683,LewinEytan2004,soto2011}.
\end{proof}

\subsection{Approximation for the \pbs{MMRM}}

Let $S=R\cup B$ be a colored point set in the plane.
Let $\mathcal{R}_1$ and $\mathcal{R}_2$ be the next two families
of rectangles of $\mathcal{R}(S)$ (see Figure~\ref{fig:classes}): 
\begin{itemize}
\item $\mathcal{R}_1$ contains the blue rectangles with a point
of $S$ in the bottom-left corner, and the red rectangles with a point
of $S$ in the bottom-right corner.
\item $\mathcal{R}_2$ contains the blue rectangles with a point
of $S$ in the bottom-right corner, and the red rectangles with a point
of $S$ in the bottom-left corner.
\end{itemize}

\begin{figure}[h]
	\centering	
	\includegraphics[scale=0.5,page=1]{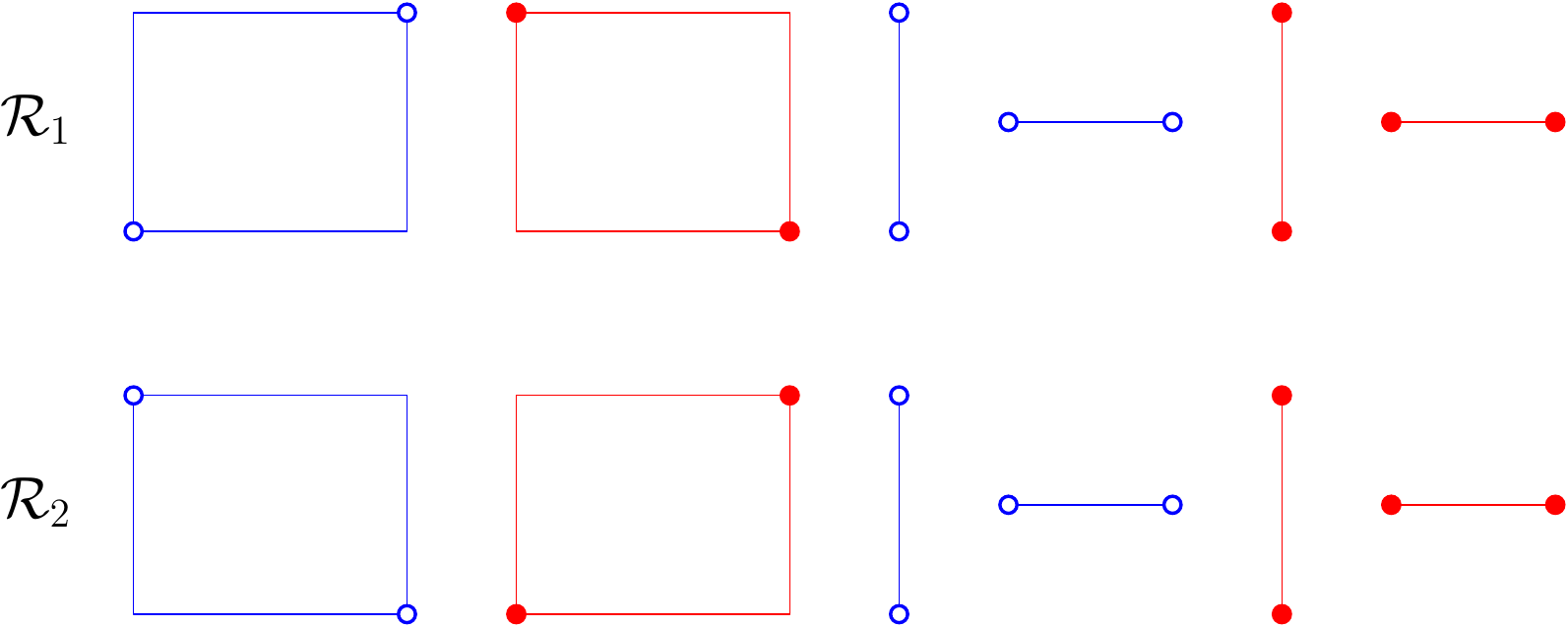}		
	\caption{\small{The families $\mathcal{R}_1$ and $\mathcal{R}_2$.}}
	\label{fig:classes}
\end{figure}

\begin{lemma}\label{lem:aprox-Ri}
There exists a polynomial-time (1/2)-approximation algorithm for the maximum independent
set of $\mathcal{R}_1$ and $\mathcal{R}_2$, respectively.
\end{lemma}

\begin{proof}
Consider the family $\mathcal{R}_1$, the arguments for the family $\mathcal{R}_2$
are analogous.
Let $\mathsf{OPT}_1$ denote the size of a maximum independent set in $\mathcal{R}_1$.
Observe that a blue and a red rectangle in $\mathcal{R}_1$ can have only
a piercing intersection, that two rectangles of the same color
cannot have a side intersection, and that $G_{p,c}(\mathcal{R}_1)$ is 
a complete set of rectangles on $S$.
Let $H$ be a maximum independent set of $G_{p,c}(\mathcal{R}_1)$ which can be found
in polynomial time by Lemma~\ref{lem:solve-pierce-misr}.
Note that in $H$ every blue rectangle is independent from
every red rectangle, and rectangles of the same color can have
point intersections only. Further observe that the graph $G(H)$
is acyclic and thus 2-coloreable. Such a 2-coloring of $G(H)$ 
can be found in polynomial time and gives an independent set $I$ of $H$ with at least
$|H|/2$ rectangles, which is an independent set in $\mathcal{R}_1$.
The set $I$ is the approximation and satisfies
$\mathsf{OPT}_1\leq |H|\leq 2|I|$. The result thus follows.\seacabo
\end{proof}

\begin{theorem}\label{lem:aprox-MMRM}
There exists a polynomial-time (1/4)-approximation algorithm for the \pbs{MMRM}.
\end{theorem}

\begin{proof}
Let $\mathsf{OPT}$ denote the size of a maximum independent set in $\mathcal{R}(S)$,
and $\mathsf{OPT}_1$ and $\mathsf{OPT}_2$ denote the sizes of the maximum 
independent sets in $\mathcal{R}_1$ and $\mathcal{R}_2$, respectively.
Let $I_1$ be a (1/2)-approximation for the maximum independent set in $\mathcal{R}_1$
and $I_2$ be a (1/2)-approximation for the maximum independent set in $\mathcal{R}_2$
(Lemma~\ref{lem:aprox-Ri}).
The approximation for the \pbs{MMRM} is to return the set with maximum elements
between $I_1$ and $I_2$.
Since $\mathsf{OPT}\leq \mathsf{OPT}_1+\mathsf{OPT}_2\leq 2|I_1|+2|I_2|\leq 4\max\{|I_1|,|I_2|\}$,
the result follows.\seacabo
\end{proof}

\subsection{Approximation for the \pbs{MMRM}}

Let $S=R\cup B$ be a colored point set in the plane.
Let $\overline{\mathcal{R}}_1$, $\overline{\mathcal{R}}_2$,
$\overline{\mathcal{R}}_3$, and $\overline{\mathcal{R}}_4$ be the next four families
of rectangles of $\overline{\mathcal{R}}(S)$: 
\begin{itemize}
\item $\overline{\mathcal{R}}_1$ contains the rectangles with a blue point
in the bottom-left corner.
\item $\overline{\mathcal{R}}_2$ contains the rectangles with a red point
in the bottom-left corner.
\item $\overline{\mathcal{R}}_3$ contains the rectangles with a blue point
in the bottom-right corner.
\item $\overline{\mathcal{R}}_4$ contains the rectangles with a red point
in the bottom-right corner.
\end{itemize}
Each of the above four families are complete sets of rectangles on $S$,
where every two rectangles have either a corner or a piercing intersection.
Then the maximum independent set in each family can be found in polynomial time (Lemma~\ref{lem:solve-pierce-misr}).
These observations imply the next result:

\begin{theorem}\label{lem:aprox-MBRM}
There exists a polynomial-time (1/4)-approximation algorithm for the \pbs{MBRM}.
\end{theorem}

\section{Hardness}\label{sec:hardness} 

In this section we prove that the \pbs{MMRM} and the \pbs{MBRM} are
$\mathsf{NP}$-hard even if further conditions are assumed. To this end
we consider the next decision problems:

\medskip

\noindent\pb{Perfect Monochromatic Rectangle Matching ($\mathsf{PMRM}$)}: {\em Is there a perfect monochromatic 
strong matching of $S$ with rectangles?}

\medskip

\noindent\pb{Perfect Bichromatic Rectangle Matching ($\mathsf{PBRM}$)}: {\em Is there a perfect bichromatic 
strong matching of $S$ with rectangles?}
 
\medskip

\noindent Proving that the \pbs{PMRM} and the \pbs{PBRM} are $\mathsf{NP}$-complete,
even on certain additional conditions,
implies that the \pbs{MMRM} and the \pbs{MBRM} are $\mathsf{NP}$-hard under
the same conditions.

In our proofs we use a reduction from the \pb{Planar 1-in-3 SAT} 
which is $\mathsf{NP}$-complete~\cite{MulzerR08}.
The input of the \pb{Planar 1-in-3 SAT}
is a Boolean formula in 3-CNF whose associated
graph is planar, and the formula is accepted if and only if there
exists an assignment to its variables such that in each clause
exactly one literal is satisfied~\cite{MulzerR08}. Given any
planar 3-SAT formula, our main idea is to construct a point set $S=S_1\cup S_2$,
such that: the elements of $S_2$ force to match certain pairs 
of points in $S_1$ and those pairs can only be matched with (axis-aligned) segments, 
there always exists a perfect matching with segments for $S_2$ independently of $S_1$,
and there exists a perfect matching with segments for $S_1$ independently of $S_2$ if and only 
the formula is accepted.

The above method can be applied in the construction of
Kratochv{\'\i}l and Ne{\v{s}}et{\v{r}}il~\cite{kratochvil1990} that proves
that finding a maximum independent set in a family of axis-aligned segments is $\mathsf{NP}$-hard. 
Indeed, we can put the elements of $S_1$ at the endpoints of the segments $T$ of their construction,
by first modelling the parallel overlapping segments by segments sharing an endpoint. 
Then the elements of $S_2$ are added in an way that every two elements of $S_1$ can be 
matched if and only if they are endpoints of
the same segment in $T$. This approach would give us a prove that
our optimization problems are $\mathsf{NP}$-hard, but not that our
perfect matching decision problems are $\mathsf{NP}$-complete
which are stronger results. 
On the other hand, our hardness proofs
give and alternative $\mathsf{NP}$-hardness proof for the problem of 
finding a maximum independent set in a family of axis-aligned segments~\cite{kratochvil1990}.

\begin{theorem}\label{theo:hardness}
The \pbs{PMRM} is $\mathsf{NP}$-complete, even if we restrict the matching
rectangles to segments.
\end{theorem}

\begin{proof}
Given a combinatorial matching of $S$, certifying that such a matching
is monochromatic, strong, and perfect can be done in  polynomial time.
Then the \pbs{PMRM} is in $\mathsf{NP}$. We prove now that
the \pbs{PMRM} is $\mathsf{NP}$-hard.

Let $\varphi$ be a planar 3-SAT formula. 
The (planar) graph associated with $\varphi$ can 
be represented in the plane as in Figure~\ref{fig:planar3SAT-sample},
where all variables lie on an horizontal line,
and all clauses are represented by {\em non-intersecting} three-legged combs~\cite{Knuth1992}.
Using this embedding, which can be constructed in a grid of polynomial size~\cite{Knuth1992},
we construct a set $S$ of red and blue integer-coordinate points in a polynomial-size grid,
such that there exists a perfect monochromatic strong matching with (axis-aligned) segments in $S$
if and only if $\varphi$ is accepted.

\begin{figure}[h]
	\centering	
	\includegraphics[scale=1]{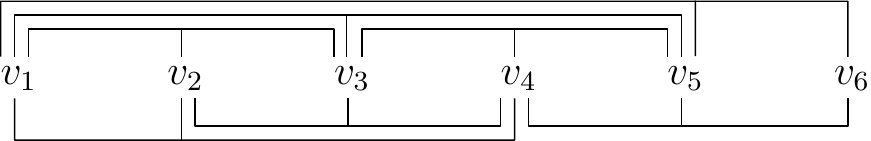}		
	\caption{\small{Planar representation of $\varphi=(v_1 \lor \overline{v_2} \lor v_3) 
	\land (v_3 \lor \overline{v_4} \lor \overline{v_5})\land(\overline{v_1} \lor \overline{v_3} \lor v_5)
	\land(v_1 \lor \overline{v_2} \lor v_4)\land(\overline{v_2} \lor \overline{v_3} \lor \overline{v_4})
	\land(\overline{v_4} \lor v_5 \lor \overline{v_6})\land(\overline{v_1} \lor v_5 \lor v_6)$.}}
\label{fig:planar3SAT-sample}
\end{figure}

For an overview of our construction of $S$, refer to Figure~\ref{fig:clause-eval}. 
We use variable gadgets (the dark-shaded rectangles called variable rectangles) 
and clause gadgets (the light-shaded orthogonal polygon representing the three-legged comb).

\begin{figure}[h]
	\centering	
	\includegraphics[scale=0.9,page=1]{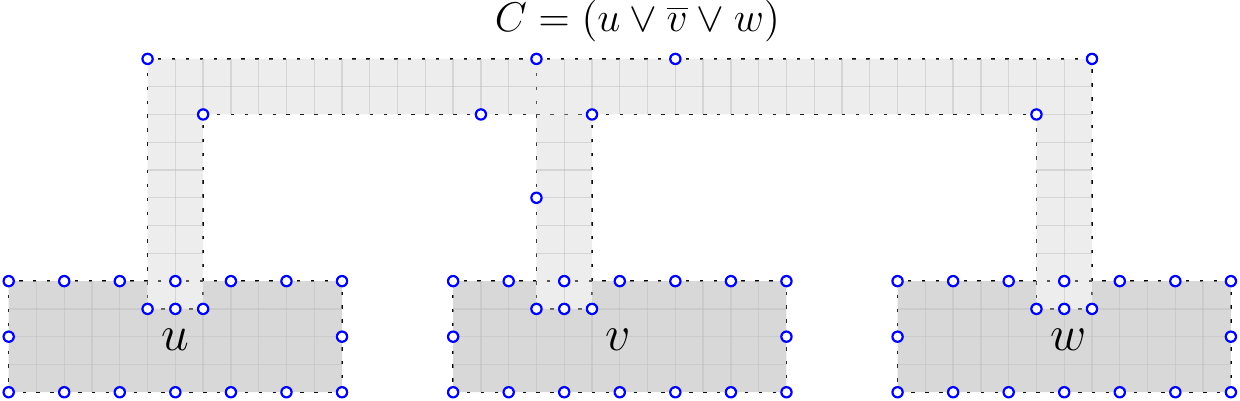}		
	\caption{\small{The variable gadgets and the clause gadgets. In the figure, each variable $u$, $v$, $w$
	might participate in other clauses.}}
	\label{fig:clause-eval}
\end{figure}

\paragraph*{Variable gadgets:}
For each variable $v$, its rectangle $Q_v$ has height $4$ and width $6\cdot d(v)$,
where $d(v)$ is the number of clauses in which $v$ appears. We assume that each
variable appears in every clause at most once. Along the boundary
of $Q_v$, starting from a vertex, we put blue points so that every two successive points are at distance
$2$ from each other. We number consecutively in clockwise order these $4+6\cdot d(v)$ points, starting from the 
top-left vertex of $Q_v$ which is numbered 1.

\paragraph*{Clause gadgets:}
Let $C$ be a clause with variables $u$, $v$, and $w$, appearing in this order
from left to right in the embedding of $\varphi$. Assume w.l.o.g.\ that
the gadget of $C$ is above the horizontal line through the variables.
Every leg of the gadget of $C$ overlaps the rectangles of its 
corresponding variable (denoted $x$) in a rectangle $Q_{x,C}$ of height 1 and width 2, 
so that the midpoint of the top side of $Q_{x,C}$ is a blue point in the boundary of $Q_x$.
The overlapping satisfies that such a midpoint is numbered with an even number
if and only if $x$ appears positive in $C$.
We further put three blue points equally spaced in the bottom side of $Q_{x,C}$, and 
other 9 blue points in the boundary of the gadget, as shown in Figure~\ref{fig:clause-eval}.

\paragraph*{Forcing a convenient matching of the blue points:}
We add red points (a polynomial number of them) in such a way that any
two blue points $a$ and $b$ can be matched if and only if $D(a,b)$ is a segment 
of any dotted line and does not contain any other colored point than $a$ and $b$
(see Figure~\ref{fig:clause-eval}). This can be done as follows: Since blue points have all
integer coordinates, we can scale the blue point set (multiplying by 2) so that every element
has even $x$- and $y$-coordinates. Then, we put a red point over every point of at least one 
odd coordinate that is not over any dotted line. We finally scale again the points, 
the blue and the red ones,
and make a copy of the scaled red points and move it one unit downwards.
%

\paragraph*{Reduction:}
Observe that in each variable $v$, the blue points along the boundary of $Q_v$ can
be matched independently of the other points, and that they have two perfect strong
matchings: the {\em 1-matching} that matches the $i$th point with the $(i+1)$th point for all odd $i$; 
and the {\em 0-matching} that matches the $i$th point with the $(i+1)$th one for all even $i$. 
In each clause $C$ in which $v$ appears, each of these two matchings {\em forces} a maximum 
strong matching on the blue points 
in the leg of the gadget of $C$ that overlaps $Q_v$, until reaching the points
in the union of the three legs. 
We consider that variable $v=1$
if we use the 1-matching, and consider $v=0$
if the 0-matching is used.
Let $C$ be a clause with variables $u$, $v$, and $w$;
and draw perfect strong matchings on the blue points of the boundaries
of $Q_u$, $Q_v$, and $Q_w$, respectively, giving values to $u$, $v$, and $w$.
%
%
Notice that if exactly one among $u$, $v$, and $w$ makes $C$ positive, 
then the strong matching forced in the blue points of the gadget of $C$ is perfect 
(see Figure~\ref{fig:clause-eval-true-1u} and Figure~\ref{fig:clause-eval-true-1v}).
Otherwise, if none or at least two among $u$, $v$, and $w$ make $C$ positive, then
the strong matching forced on the blue points of the gadget of $C$ is not perfect since
at least 2 blue points are unmatched (see Figure~\ref{fig:clause-eval-true-0} and
Figure~\ref{fig:clause-eval-true-2}).
Finally, note that the red points admit a perfect
strong matching with segments such that no segment contains a blue point.
Therefore, we can ensure that the 3-SAT formula $\varphi$ can be accepted
if and only if the point set $S$ admits a perfect strong matching with segments.\seacabo
\end{proof}

\begin{figure}[h]
	\centering	
	\includegraphics[scale=0.9,page=3]{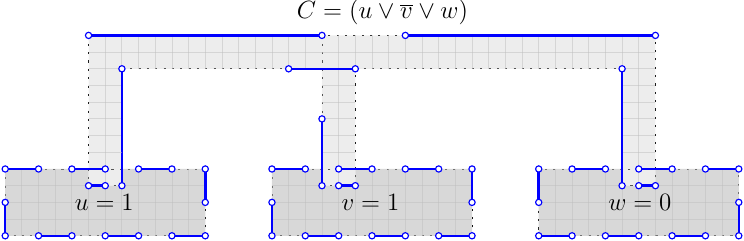}		
	\caption{\small{If $u=1$, $v=1$, and $w=0$, then only $u$ makes $C$ positive
	and there exists a perfect strong matching on the blue points.}}
	\label{fig:clause-eval-true-1u}
\end{figure}

\begin{figure}[h]
	\centering	
	\includegraphics[scale=0.9,page=4]{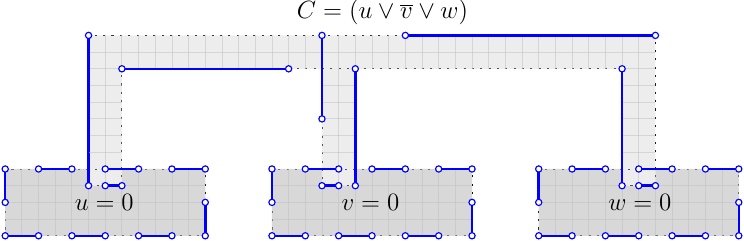}		
	\caption{\small{If $u=0$, $v=0$, and $w=0$, then only $v$ makes $C$ positive
	and there exists a perfect strong matching on the blue points.}}
	\label{fig:clause-eval-true-1v}
\end{figure}

\begin{figure}[h]
	\centering	
	\includegraphics[scale=0.9,page=2]{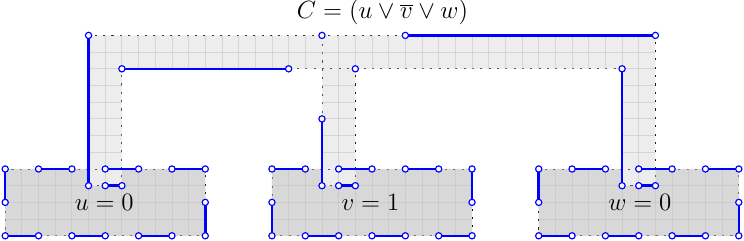}		
	\caption{\small{If $u=0$, $v=1$, and $w=0$, then no variable makes $C$ positive
	and there does not exist any perfect strong matching on the blue points.}}
	\label{fig:clause-eval-true-0}
\end{figure}

\begin{figure}[h]
	\centering	
	\includegraphics[scale=0.9,page=5]{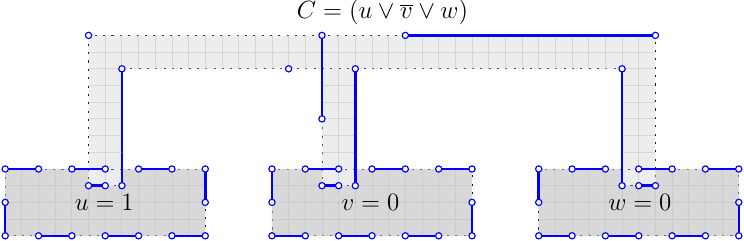}		
	\caption{\small{If $u=1$, $v=0$, and $w=1$, then two variables make $C$ positive
	and there does not exist any perfect strong matching on the blue points.}}
	\label{fig:clause-eval-true-2}
\end{figure}

%

Suppose now that the two-colored point set $S$
is in general position. In what follows we show that the 
\pbs{PMRM} remains $\mathsf{NP}$-complete under this assumption.
To this end we first perturb the two-colored point set of the construction
of the proof of Theorem~\ref{theo:hardness} so that no two points 
share the same $x$- or $y$-coordinate, and second show that two
points of $S$ can be matched in the perturbed point set if and only if
they can be matched in the original one.

Alliez et al.~\cite{alliez97} proposed the transformation that
replaces each point $p=(x,y)$ by the point 
$\lambda(p):=((1+\varepsilon)x+\varepsilon^2 y,\varepsilon^3x+y)$ for some small enough
$\varepsilon>0$, with the aim of removing the degeneracies in a point
set for computing the Delaunay triangulation under the $L_{\infty}$ metric.
Although this transformation can be used for our purpose, by using the fact that
the points in the proof of Theorem~\ref{theo:hardness} belong to a grid 
$[0..N]^2$, where $N$ is polynomially-bounded, we use the simpler  
transformation $\lambda(p):=((1+\varepsilon)x+\varepsilon y,\varepsilon x+(1+\varepsilon)y)$ 
for $\varepsilon=1/(2N+1)$, which is linear in $\varepsilon$. Both transformations
change the relative positions of the initial points in the manner showed in 
Figure~\ref{fig:perturbation}.
Some useful properties of our transformation, stated in the next lemma,
were not stated by Alliez et al.~\cite{alliez97}.

\begin{figure}
	\centering	
	\includegraphics[scale=0.9]{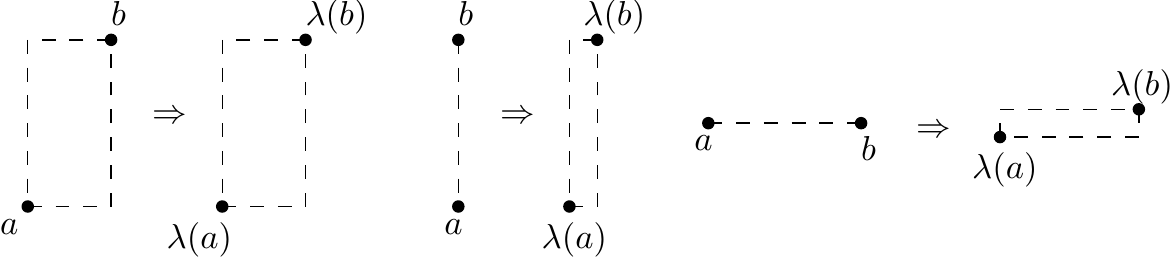}		
	\caption{\small{Perturbation of the point set to put $S$ in general position.}}
	\label{fig:perturbation}
\end{figure}

\begin{lemma}\label{lem:perturb}
Let $N$ be a natural number and $P\subseteq [0..N]^2$. 
The function $\lambda:P\rightarrow\mathbb{Q}^2$ such that 
$$\lambda(p)=\Biggl(x(p)+\frac{x(p)+y(p)}{2N+1},y(p)+\frac{x(p)+y(p)}{2N+1}\Biggr)$$
satisfies the next properties:
\begin{itemize}
\item[$(a)$] $\lambda$ is injective and the point set $\lambda(P):=\{\lambda(p):p\in P\}$
is in general position.

\item[$(b)$] For every two distinct points $a,b\in P$ 
such that $x(a)=x(b)$ or $y(a)=y(b)$, we have that
$D(a,b)\cap P=\{a,b\}$ if and only if 
$D(\lambda(a),\lambda(b))\cap \lambda(P)=\{\lambda(a),\lambda(b)\}$.

\item[$(c)$] For every three distinct points $a,b,c\in P$ 
such that $x(a)\neq x(b)$ and $y(a)\neq y(b)$, we have that
$c$ belongs to the interior of $D(a,b)$ if and only if 
$\lambda(c)$ belongs to the interior of $D(\lambda(a),\lambda(b))$.


\end{itemize}
\end{lemma}

\begin{proof}
Properties (a-c) are a consequence of $0\leq \frac{x(p)+y(p)}{2N+1} \leq \frac{2N}{2N+1}<1$.
\seacabo
\end{proof}

\begin{theorem}\label{theo:hardness-generalpos}
The \pbs{PMRM} remains $\mathsf{NP}$-complete on point sets in
general position.
\end{theorem}

\begin{proof}
Let $S$ be the colored point set generated in the reduction of the proof
of Theorem~\ref{theo:hardness}. Let $N$ be a polynomially-bounded natural
number such that $S\subset [0..N]^2$, and let $S':=\lambda(S)$, where
$\lambda$ is the function of Lemma~\ref{lem:perturb}. Consider
the next observations:
\begin{itemize}
\item[(a)] If $a,b\in S$ are red points that can be matched in $S$
because $x(a)=x(b)$ and $y(b)=y(a)-1$, then $\lambda(a)$ and $\lambda(b)$
can also be matched in $S'$ (Property (b) of Lemma~\ref{lem:perturb}).

\item[(b)] If $a,b\in S$ are blue points that can be matched in $S$,
then we have that either $x(a)=x(b)$ or $y(a)=y(b)$, which implies that
$\lambda(a)$ and $\lambda(b)$ can also be matched in $S'$ by
Property (b) of Lemma~\ref{lem:perturb}.

\item[(c)] If $a,b\in S$ are blue points that cannot be matched in $S$ 
because $D(a,b)$ is a segment containing a red point $c\in S$, then 
neither $\lambda(a)$ and $\lambda(b)$ can
be matched in $S'$ (Property (b) of Lemma~\ref{lem:perturb}).

\item[(d)] If $a,b\in S$ are blue points that cannot be matched in $S$ 
because $D(a,b)$ is a box containing
a point $c\in S$ in the interior, then neither $\lambda(a)$ and $\lambda(b)$ can
be matched in $S'$ since the box $D(\lambda(a),\lambda(b))$ contains $\lambda(c)$
(Property (c) of Lemma~\ref{lem:perturb}).
\end{itemize}
The above observations imply that there exists a perfect strong rectangle matching in
$S$ if and only if it exists in $S'$.
The result thus follows since $S'$ is in general position by Property (a) of Lemma~\ref{lem:perturb}.
\seacabo
\end{proof}

Combining the construction of Theorem~\ref{theo:hardness} with
the perturbation of Lemma~\ref{lem:perturb}, we can prove that
the \pbs{PMRM} is also $\mathsf{NP}$-complete when
all points have the same color, and that
the \pbs{PBRM} is also $\mathsf{NP}$-complete.

\begin{lemma}\label{lem:blocking}
Let $M_1:=\{(0,0),(5,0),(5,5),(0,5)\}$ and $M_2:=\{(1,3),(2,2),(2,3),$ $(2,4),(3,1),(3,2),(3,3),(4,2)\}$ 
be two point sets. The point set $M_1\cup M_2$ has a perfect strong matching with rectangles,
and for every proper subset $M'_1\subset M_1$ the point set $M'_1\cup M_2$ does not have
any perfect strong matching with rectangles.
\end{lemma}

\begin{proof}
The proof is straightforward (see Figure~\ref{fig:green-x}, Figure~\ref{fig:green-y},
and Figure~\ref{fig:green-z}).
\end{proof}

\begin{figure}[h]
	\centering
	\subfloat[]{
		\includegraphics[scale=0.8,page=4]{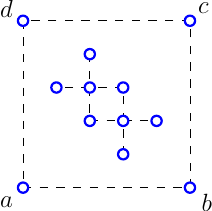}
		\label{fig:green-x}
	}\hspace{0.2cm}
	\subfloat[]{
		\includegraphics[scale=0.8,page=5]{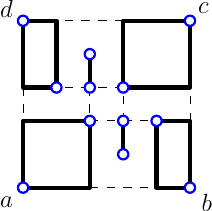}
		\label{fig:green-y}
	}\\
	\subfloat[]{
		\includegraphics[scale=0.8,page=6]{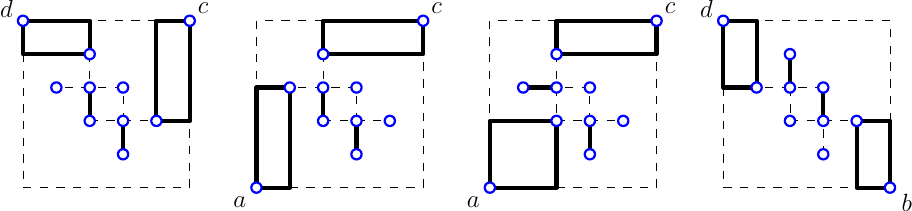}
		\label{fig:green-z}
	}
	\caption{\small{(a) The point set $M_1\cup M_2$. (b) A perfect matching of $M_1\cup M_2$.
	(c) If exactly two points among $a,b,c,d$ are removed, then the remaining points do not have
	any perfect strong matching. 
	}}
\end{figure}

\begin{theorem}\label{theo:hardness-bichromatic}
The \pbs{PMRM} remains $\mathsf{NP}$-complete if all elements
of $S$ have the same color.
\end{theorem}

\begin{proof}
Let $R_0$ and $B_0$ be the sets of the red points and the blue points,
respectively, in the proof of Theorem~\ref{theo:hardness}.
Let $Q$ be a set of (artificial) {\em green} points to block
the forbidden matching rectangles in $B_0$, that is,
for every two points $p,q\in B_0$ we have that $D(p,q)$ does not contain
elements of $B_0\cup Q$ other than $p$ and $q$ if and only if $D(p,q)$ is
a matching rectangle in $R_0\cup B_0$. 
In other words, $p,q$ can be matched in $R_0\cup B_0$ if and only
if they can be matched in $B_0\cup Q$.
The point set $S_1:=B_0\cup Q$
belongs to the grid $[0..N]^2$, where $N$ is polynomially-bounded, and is not 
in general position. Let $S_2:=\lambda(S_1)$, where $\lambda$ is the function
of the Lemma~\ref{lem:perturb}. We now replace each green point $g$ of $Q$ by a 
translated and stretched copy $S_g$ of the set $M_1\cup M_2$ of Lemma~\ref{lem:blocking},
with all elements colored blue
(see Figure~\ref{fig:green-x}). 
Let $S:=B_0\cup (\bigcup_{g\in Q}S_g)$. 
Putting the elements of $S_g$ close enough one another for every $g$, we can guarantee
that if we want to obtain a perfect strong matching in $S$ then we must have by Lemma~\ref{lem:blocking}
a perfect strong matching in each $S_g$ in particular (see Figure~\ref{fig:green-y})
Therefore, the set $S_g$ acts as the green point $g$ blocking
the forbidden matching rectangles in $B_0$. 
The construction of $S$ starts from the planar 3-SAT formula $\varphi$ 
of the proof of Theorem~\ref{theo:hardness}, and using all the above arguments
we can claim that there exists a perfect strong matching in $S$ if and
only if the formula $\varphi$ is accepted. Hence, the \pbs{PMRM} with input points of the same color
is $\mathsf{NP}$-complete
since there exists a polynomial-time reduction from the
\pb{Planar 1-in-3 SAT}.\seacabo
\end{proof}

\begin{theorem}\label{theo:hardness-bichromatic}
The \pbs{PBRM} is $\mathsf{NP}$-complete, even if the point set $S$
is in general position.
\end{theorem}

\begin{proof}
Let $R_0$ and $B_0$ be the sets of the red points and the blue points,
respectively, in the proof of Theorem~\ref{theo:hardness}.
Change to color red elements of $B_0$, to obtain the colored point set $S_0$, so that 
for every segment matching two blue points in $R_0\cup B_0$ exactly one of the matched points is 
changed to color red (see Figure~\ref{fig:clause-bichromatic}). 
For every point $p\in B_0$, let $p'$ denote the corresponding point in $S_0$, and vice versa.
Let $Q$ be a set of (artificial) {\em green} points to block
the forbidden matching rectangles in $S_0$, that is,
for every two points $p',q'\in S_0$ we have that $D(p',q')$ does not contain
elements of $S_0\cup Q$ other than $p'$ and $q'$ if and only if $D(p,q)$ is
a matching rectangle in $R_0\cup B_0$. The point set $S_1:=S_0\cup Q$
belongs to the grid $[0..N]^2$, where $N$ is polynomially-bounded, and is not 
in general position. Let $S_2:=\lambda(S_1)$, where $\lambda$ is the function
of the Lemma~\ref{lem:perturb}. We now replace each green point $g$ of $Q$ by the 
set $S_g$ of eight red and blue points in general position (see Figure~\ref{fig:green-a}). 
Let $S:=S_0\cup (\bigcup_{g\in Q}S_g)$. 
Putting the elements of $S_g$ close enough one another for every $g$, we can guarantee
that $S$ is also in general position and that for every $g$ the points of $S_g$ 
appear together in both the left-to-right and the
top-down order of $S$. This last condition ensures that if we want to obtain a 
perfect strong matching in $S$ then we must have a perfect
strong matching for each $S_g$ in particular (see Figure~\ref{fig:green-b} and Figure~\ref{fig:green-c})
because for all $g$ every red point of $S_g$ cannot be matched with any blue point not in $S_g$.
Therefore, the set $S_g$ acts as the green point $g$ blocking
the forbidden matching rectangles in $S_0$. 
%
Hence, the \pbs{PBRM} is $\mathsf{NP}$-complete,
even on points in general position.
\seacabo
\end{proof}

\begin{figure}[h]
	\centering
	\subfloat[]{
		\includegraphics[scale=0.9,page=6]{variable_clause2.pdf}
		\label{fig:clause-bichromatic}
	}\\
	\subfloat[]{
		\includegraphics[scale=0.6,page=1]{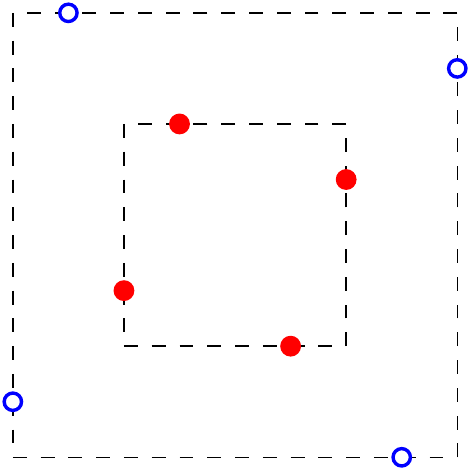}
		\label{fig:green-a}
	}\hspace{0.2cm}
	\subfloat[]{
		\includegraphics[scale=0.6,page=2]{green_point.pdf}
		\label{fig:green-b}
	}
	\hspace{0.2cm}
	\subfloat[]{
		\includegraphics[scale=0.6,page=3]{green_point.pdf}
		\label{fig:green-c}
	}
	\caption{\small{Proof of Theorem~\ref{theo:hardness-bichromatic}.
	(a) Changing the colors of the blue points in the gadgets of the proof
	of Theorem~\ref{theo:hardness}. (b) The eight points (close enough one another) 
	that replace each green point. (c) One of the only two ways to match the points
	corresponding to a green point in order to obtain a perfect matching. (d) The other way.}}
\end{figure}

\section{Discussion}\label{sec:discussion}

We have proved that finding a maximum strong matching of a two-colored
point set, with either rectangles containing points from the same color
or rectangles containing points of different colors, is $\mathsf{NP}$-hard
and provide a $(1/4)$-approximation for each case. 
Our approximation algorithms provide a (1/4)-approximation for the problem of 
finding a maximum strong rectangle matching
of points of the same color, studied by Bereg et al.~\cite{bereg2009}. However,
the approximation ratio is smaller than $2/3$, the one given by Bereg et al.
We leave as open to find a better constant approximation algorithm for our problems, 
or a $\mathsf{PTAS}$.
On the other hand, finding a constant-approximation algorithm for the general case of
the \pb{Maximum Independent Set of Rectangles} is still an intriguing 
open question.

\small

\bibliographystyle{plain}
\bibliography{matching}

\begin{thebibliography}{10}

\bibitem{abrego2009}
B.~M. \'Abrego, E.~M. Arkin, S.~Fern\'andez-Merchant, F.~Hurtado, M.~Kano,
  J.~S.B. Mitchell, and J.~Urrutia.
\newblock Matching points with squares.
\newblock {\em Discrete \& Computational Geometry}, 41(1):77--95, 2009.

\bibitem{adamaszekW13}
A.~Adamaszek and A.~Wiese.
\newblock Approximation schemes for maximum weight independent set of
  rectangles.
\newblock {\em CoRR}, abs/1307.1774, 2013.

\bibitem{Agarwal200683}
P.~K. Agarwal and N.~H. Mustafa.
\newblock Independent set of intersection graphs of convex objects in 2d.
\newblock {\em Computational Geometry}, 34(2):83 -- 95, 2006.

\bibitem{AgarwalKS98}
P.~K. Agarwal, M.~J. van Kreveld, and S.~Suri.
\newblock Label placement by maximum independent set in rectangles.
\newblock {\em Computational Geometry}, 11(3-4):209--218, 1998.

\bibitem{alliez97}
P.~Alliez, O.~Devillers, and J.~Snoeyink.
\newblock Removing degeneracies by perturbing the problem or the world.
\newblock Technical Report 3316, INRIA, 1997.

\bibitem{bereg2009}
S.~Bereg, N.~Mutsanas, and A.~Wolff.
\newblock Matching points with rectangles and squares.
\newblock {\em Comput. Geom. Theory Appl.}, 42(2):93--108, 2009.

\bibitem{Chalermsook2011}
P.~Chalermsook.
\newblock Coloring and maximum independent set of rectangles.
\newblock In {\em Approximation, Randomization, and Combinatorial Optimization.
  Algorithms and Techniques}, volume 6845 of {\em LNCS}, pages 123--134.
  Springer Berlin Heidelberg, 2011.

\bibitem{Chalermsook2009}
P.~Chalermsook and J.~Chuzhoy.
\newblock Maximum independent set of rectangles.
\newblock In {\em Proc. of the twentieth Annual ACM-SIAM Symposium on Discrete
  Algorithms}, SODA'09, pages 892--901, Philadelphia, USA, 2009.

\bibitem{Chan03}
T.~M. Chan.
\newblock Polynomial-time approximation schemes for packing and piercing fat
  objects.
\newblock {\em J. Algorithms}, 46(2):178--189, 2003.

\bibitem{Chan200419}
T.~M. Chan.
\newblock A note on maximum independent sets in rectangle intersection graphs.
\newblock {\em Information Processing Letters}, 89(1):19 -- 23, 2004.

\bibitem{chan2009}
T.~M. Chan and S.~Har-Peled.
\newblock Approximation algorithms for maximum independent set of pseudo-disks.
\newblock In {\em Proceedings of the Twenty-fifth Annual Symposium on
  Computational Geometry}, SCG '09, pages 333--340, 2009.

\bibitem{Dumitrescu2001}
A.~Dumitrescu and R.~Kaye.
\newblock Matching colored points in the plane: Some new results.
\newblock {\em Computational Geometry}, 19(1):69 -- 85, 2001.

\bibitem{Dumitrescu2000}
A~Dumitrescu and W.~L. Steiger.
\newblock On a matching problem in the plane.
\newblock {\em Discrete Mathematics}, 211:183--195, 2000.

\bibitem{ErlebachJS05}
T.~Erlebach, K.~Jansen, and E.~Seidel.
\newblock Polynomial-time approximation schemes for geometric intersection
  graphs.
\newblock {\em SIAM J. Comput.}, 34(6):1302--1323, 2005.

\bibitem{Fowler1981}
R.~J. Fowler, M.~Paterson, and S.~L. Tanimoto.
\newblock Optimal packing and covering in the plane are {NP}-complete.
\newblock {\em Inf. Process. Lett.}, 12(3):133--137, 1981.

\bibitem{Grotschel1984325}
M.~Gr{\"{o}}tschel, L.~Lov\'asz, and A.~Schrijver.
\newblock Polynomial algorithms for perfect graphs.
\newblock In {\em Topics on Perfect Graphs}, volume~88, pages 325 -- 356. 1984.

\bibitem{Imai1983}
H.~Imai and T.~Asano.
\newblock Finding the connected components and a maximum clique of an
  intersection graph of rectangles in the plane.
\newblock {\em J. Alg.}, 4(4):310--323, 1983.

\bibitem{KhannaMP98}
S.~Khanna, S.~Muthukrishnan, and M.~Paterson.
\newblock On approximating rectangle tiling and packing.
\newblock In {\em ACM-SIAM Symposium on Discrete Algorithms}, pages 384--393,
  1998.

\bibitem{Knuth1992}
D.~E. Knuth and A.~Raghunathan.
\newblock The problem of compatible representatives.
\newblock {\em SIAM J. Discret. Math.}, 5(3):422--427, 1992.

\bibitem{kratochvil1990}
J.~Kratochv{\'\i}l and J.~Ne{\v{s}}et{\v{r}}il.
\newblock Independent set and clique problems in intersection-defined classes
  of graphs.
\newblock {\em Commen. Math. Univ. Carolinae}, 31(1):85--93, 1990.

\bibitem{larson1990}
L.~C. Larson.
\newblock {\em Problem-solving through problems}.
\newblock Problem books in mathematics. Springer, 1990.

\bibitem{LewinEytan2004}
L.~Lewin-Eytan, J.~Naor, and A.~Orda.
\newblock Admission control in networks with advance reservations.
\newblock {\em Algorithmica}, 40(4):293--304, 2004.

\bibitem{MulzerR08}
W.~Mulzer and G.~Rote.
\newblock Minimum-weight triangulation is {NP}-hard.
\newblock {\em J. ACM}, 55(2), 2008.

\bibitem{rim1995}
C.~S. Rim and K.~Nakajima.
\newblock On rectangle intersection and overlap graphs.
\newblock {\em IEEE Transactions on Circuits and Systems}, 42(9):549--553,
  1995.

\bibitem{soto2011}
J.~A. Soto and C.~Telha.
\newblock Jump number of two-directional orthogonal ray graphs.
\newblock In {\em Integer Programming and Combinatorial Optimization}, volume
  6655 of {\em LNCS}, pages 389--403. Springer Berlin Heidelberg, 2011.

\end{thebibliography}

\end{document}